\documentclass[a4paper]{article}

\usepackage[disable]{todonotes}

\usepackage[utf8]{inputenc}
\usepackage{xcolor}

\newcommand{\lspan}{\mathsf{span}}

\usepackage{fullpage}
\usepackage{parskip}
\usepackage{amsmath}
\usepackage{mathtools}
\usepackage{amssymb}
\usepackage{tabularx}
\usepackage[toc,page]{appendix}
\usepackage{xspace}

\usepackage{microtype}
\usepackage{bbm}
\usepackage{amsthm}
\usepackage{enumerate}
\usepackage[]{hyperref}

\newcommand{\norm}[1]{\ensuremath{\left\|#1\right\|}}
\usepackage{comment}
\usepackage[classfont=sanserif,langfont=sanserif,funcfont=sanserif]{complexity}
\usepackage{thmtools, thm-restate}
\usepackage{xcolor}

\newtheorem{definition}{Definition}
\newtheorem*{theorem*}{Theorem}
\newtheorem{theorem}{Theorem}

\newtheorem{observation}[theorem]{Observation}

\newtheorem{lemma}{Lemma}

\newtheorem*{conjecture*}{Conjecture}
\newtheorem{corollary}{Corollary}
\declaretheorem[name=Theorem]{thm}

\newcommand{\ket}[1]{|{#1}\rangle}

\newcommand{\cH}{{\mathcal H}}
\newcommand{\cU}{{\mathcal U}}
\newcommand{\cB}{{\mathcal B}}
\newcommand{\cQ}{{\mathcal Q}}
\newcommand{\cG}{{\mathcal G}}
\newcommand{\cI}{{\mathcal I}}

\newcommand{\RAG}{\mathsf{RAG}}
\newcommand{\CNOT}{\mathsf{CNOT}}
\newcommand{\APPLY}{\mathsf{APPLY}}

\newcommand{\READ}{\mathsf{READ}}

\newcommand{\QRAM}{\mathsf{QRAM}}

\newcommand{\circuitDesc}[1]{\mathcal{#1}}

\newcommand{\unitarySwap}{U_{swap}}
\newcommand{\unitaryLook}{U_{lookup}}
\newcommand{\unitaryToggle}{U_{toggle}}
\newcommand{\controlledToggle}{C_{toggle}}

\newcommand{\unitarySuperpose}{U_{\text{superpose}}}
\newcommand{\unitaryAlloc}{U_{\text{alloc}}}
\newcommand{\unitaryFree}{U_{\text{free}}}

\newcommand{\QRT}[1]{R_Q(#1)}

\newcommand{\TCP}{\text{TCP}}
\newcommand{\closestPair}{\text{CP}}
\newcommand{\elementDistinctness}{\text{Element Distinctness}}

\newcommand{\fid}{\textnormal{id}}

\newcommand{\kedT}{T_{k}}
\newcommand{\tcpT}{T_{CP}}
\newcommand{\bitVector}{\texttt{B}}
\newcommand{\countSol}{\texttt{count}}
\newcommand{\flag}{\texttt{flag}}

\newcommand{\externalCount}{\texttt{external}}

\usepackage{mathtools}

\newcommand{\eps}{\varepsilon}
\newcommand{\ZO}{\{0,1\}}
\newcommand{\tO}{\widetilde O}

\title{Memory Compression with Quantum Random-Access Gates}
\author{Harry Buhrman${^*}$ \\ \href{mailto:harry.buhrman@cwi.nl}{harry.buhrman@cwi.nl} 
\and Bruno Loff${^\dag}$ \\ \href{mailto:bruno.loff@gmail.com}{bruno.loff@gmail.com}
\and Subhasree Patro${^*}$ \\ \href{mailto:subhasree.patro@cwi.nl}{subhasree.patro@cwi.nl}
\and Florian Speelman$^*$ \\ \href{mailto:f.speelman@uva.nl}{f.speelman@uva.nl} \\[0.35cm]
   $^*$QuSoft, University of Amsterdam, 
 CWI Amsterdam,\\$^\dag$University of Porto and INESC-Tec}

\date{\today}

\begin{document}

\maketitle
%If we can afford some extra error probability, we can do so using hashing, while preserving the time that the algorithm takes, up to constant factors. If we can afford a $\log n$ factor increase in time, we can use self-balancing trees (for example), and do so without any error whatsoever.
\begin{abstract}
    In the classical RAM, we have the following useful property. If we have an algorithm that uses $M$ memory cells throughout its execution, and in addition is sparse, in the sense that, at any point in time, only $m$ out of $M$ cells will be non-zero, then we may ``compress'' it into another algorithm which uses only $m \log M$ memory and runs in almost the same time. We may do so by simulating the memory using either a hash table, or a self-balancing tree.

    We show an analogous result for quantum algorithms equipped with quantum random-access gates. If we have a quantum algorithm that runs in time $T$ and uses $M$ qubits, such that the state of the memory, at any time step, is supported on computational-basis vectors of Hamming weight at most $m$, then it can be simulated by another algorithm which uses only $O(m \log M)$ memory, and runs in time $\tilde O(T)$.
    
    We  show how this theorem can be used, in a black-box way, to simplify the presentation in several papers. Broadly speaking, when there exists a need for a space-efficient history-independent quantum data-structure, it is often possible to construct a space-inefficient, yet sparse, quantum data structure, and then appeal to our main theorem. This results in simpler and shorter arguments.
\end{abstract}
\thispagestyle{empty}

\tableofcontents{}

\newpage
\clearpage
\setcounter{page}{1}

\section{Introduction}
\label{sec:Introduction}

This paper arose out of the authors' recent work on quantum fine-grained complexity~\cite{buhrman-3SUMhard-2021}, where we had to make use of a quantum walk, similar to how Ambainis uses a quantum walk in his algorithm for element distinctness~\cite{Ambainis-ElementDistinctness-2004}. An essential aspect of these algorithms is the use of a history-independent data-structure. In the context of our paper, we needed three slightly different data structures of this type, and on each of these occasions we saw a similar scenario. If we were only concerned with the time complexity of our algorithm, and were OK with a polynomial increase in the space complexity (the number of qubits used by the algorithm), then there was a very simple data structure that would serve our purpose. If, however, we wanted the algorithm to be space-efficient, as well, then we needed to resort to more complicated data structures.

And we made the following further observation: the simple, yet space-inefficient, data structures were actually \textit{sparse}, in the sense that although $M$ qubits were being used, all the amplitude was always concentrated on computational-basis vectors of Hamming weight $\le m \ll M$. The analogous classical scenario is an algorithm that uses $M$ memory registers, but at any time step all but $m$ of these registers are set to $0$. In the classical case, we know how to convert any such an \textit{$m$-sparse} algorithm  into an algorithm that uses $O(m \log M)$ memory, by using, e.g., a hash table. We wondered whether the same thing could be said of quantum algorithms. This turned out to be possible, and the main purpose of this paper is to explain how it can be done. We will take an arbitrary \textit{sparse} quantum algorithm, and \textit{compress} it into a quantum algorithm that uses little space.

Our main theorem is as follows (informally stated):

\begin{theorem}\label{thm:main}
Any $m$-sparse quantum algorithm using time $T$ and $M$ qubits can be simulated with $\eps$ additional error by a quantum algorithm running in time $O(T \cdot \log(\frac{T}{\eps}) \cdot \log(M))$, using $O(m \log M)$ qubits.
\end{theorem}

We will prove this result using quantum radix trees in Section~\ref{sec:CompressingUsingRadixTrees}. The result can also be proven, with slightly worse parameters, using hash tables, but we will not do so here. The sparse algorithm is allowed to use quantum random-access gates (described in the preliminaries Section \ref{sec:Definitions}), and the \textit{compressed} simulation requires such gates, even if the original algorithm does not. 

The $\log M$ factor in the time bound can be removed if we assume that certain operations on $O(\log M)$ bits can be done at $O(1)$ cost. This includes only simple operations such as comparison, addition, bitwise XOR, or swapping of two blocks of $O(\log M)$ adjacent qubits.\footnote{The qubits in each block are adjacent, but the two swapped blocks can be far apart from each other.} All these operations can be done at $O(1)$ cost in the usual classical Random-Access Machines.

The techniques used to prove our main theorem are not new: quantum radix-trees first appeared in a paper by Bernstein, Jeffery, Lange and Meurer \cite{Bernstein-Jeffery-Lange-Meurer-SubsetSum-2013} (see also Jeffery's PhD thesis \cite{Stacey-PhDThesis-2014}). One contribution of our paper is to present BJLM technique in full, as in currently available presentations of the technique, several crucial aspects of the implementation are missing or buggy\footnote{For example, some operations are defined which are not unitary. Or, there is no mention of error in the algorithms, but they actually cannot be implemented in an error-free way using a reasonable number of gates from any standard gate set.}.

But our main contribution is to use these techniques at the right level of abstraction. Theorem \ref{thm:main} is very general, and can effectively be used as a black box. One would think that Theorem \ref{thm:main}, being such a basic and fundamental statement about quantum computers, and being provable essentially by known techniques, would already be widely known. But this appears not to be the case, as papers written as recently as a year ago could be significantly simplified by appealing to such a theorem. Indeed, we believe that the use of Theorem \ref{thm:main} will save researchers a lot of work in the future, and this is our main motivation for writing this paper.

To illustrate this point, in Section \ref{sec:Simplifications} we will overview three papers \cite{Ambainis-ElementDistinctness-2004, Aaronson-ClosestPair-2019, buhrman-3SUMhard-2021} that make use of a quantum walk together with a history-independent data structure. These papers all use complicated but space-efficient data structures. As it turns out, we can replace these complicated data structures with very simple tree-like data structures. These new, simple data structures are memory inefficient but sparse, so we may then appeal to Theorem \ref{thm:main} to get similar upper bounds. The proofs become shorter: we estimate each of these papers could be cut in size by 4 to 12 pages. And furthermore, using simpler (memory inefficient but sparse) data-structures allows for a certain \textit{separation of concerns}: when one tries to describe a space-efficient algorithm, there are several bothersome details that one needs to keep track of, and they obscure the presentation of the algorithm. By using simpler data structures, these bothersome details are disappear from the proofs, and are entrusted to Theorem \ref{thm:main}.

\section{Definitions}
\label{sec:Definitions}

We let $[n] = \{1, \ldots, n \}$, and let $\binom{[n]}{\le m}$ be the set of subsets of $[n]$ of size at most $m$.

We let $\cH(N)$ denote the complex Hilbert space of dimension $N$, and we let $\cU(N)$ denote the space of unitary linear operators from $\cH(N)$ to itself (i.e.\ the unitary group). We let $\cB$ denote a set of universal quantum gates, which we will fix to containing the $\CNOT \in \cU(4)$ and all single-qubit gates, %, the Hadamard, S and T gates $H, S, T \in \cU(2)$ gates, as well as the controlled Hadamard, S, and T gates $\CH, \CS, \CT \in \cU(4)$.
but which we could have been chosen from among any of the standard possibilities. 

Of particular importance to this paper will be the set $\cQ = \cB \cup \{ \RAG_n \mid n \text{ a power of } 2 \}$ which contains our universal set together with the \textit{random-access gates}, so that $\RAG_n \in \cU(n 2^{1 + n})$ is defined on the computational basis by:
\[
\RAG_n \ket{i, b, x_0, \ldots, x_{n-1}} = \ket{i, x_i, x_0, \ldots, x_{i-1}, b, x_{i+1}, \ldots, x_{n-1}} \tag*{$\forall i\in[n], b, x_0, \ldots, x_{n-1} \in \ZO$}
\]

\medskip\noindent
We now give a formal definition of what it means to solve a Boolean relation $F \subseteq \ZO^n\times\ZO^m$ using a quantum circuit. This includes the special case when $F$ is a function. % A similar definition could be given, mutatis mutandis, for computing a family of unitaries.

A \textit{quantum circuit} over a gate set $\cG$ (such as $\cB$ or $\cQ$) is a tuple $C = (n, T, S, C_1, \ldots, C_T)$, where $T \ge 0$, $n, S \ge 1$ are natural numbers, and the $C_t$ give us a sequence of instructions. Each instruction $C_t$ comes from a set $\cI_\cG(S)$ of possible instructions, defined below. The number $n$ is the input length, the number $T$ is the \textit{time complexity}, and $S$ is the \textit{space complexity}, also called the \textit{number of wires} or the \textit{number of ancillary qubits} of the circuit. Given an input $x \in \ZO^n$, at each step $t \in \{0, \ldots, T\}$ of computation, the circuit produces an $S$-qubit state $\ket{\psi_T(x)} \in \cH(2^S)$, starting with $\ket{\psi_0(x)} = \ket{0}^{\otimes s}$, and then applying each instruction $C_t$, as we will now describe.

For each possible $q$-qubit gate $G \in \cG \cap \cU(2^q)$, and each possible ordered choice $I = (i_1, \ldots, i_{q}) \in [S]^{q}$ of distinct $q$ among $S$ qubits, we have an instruction $\APPLY_{G,I}\in \cI_\cG(S)$ which applies gate $G$ to the qubits indexed by $I$, in the prescribed order. The effect of executing the instruction $\APPLY_{G,I}$ on $\ket{\psi} \in \cH(2^S)$ is to apply $G$ on the qubits indexed by $I$, tensored with identity on the remaining $S - q$ qubits. I.e., $\APPLY_{G,I} \in \cU(2^S)$ corresponds to the unitary transformation defined on each basis state by:
\[
\APPLY_{G,I} \cdot \ket{y_I} \otimes \ket{y_J} = (G \ket{y_I}) \otimes \ket{y_J},
\]
where $J = [S]\setminus I$.

Furthermore, for each possible ordered choice $I = (i_1, \ldots, i_{\lceil \log n \rceil}) \in [S]^{\lceil \log n \rceil}$ of distinct $\lceil \log n \rceil$ among $S$ qubits, and each $i \in [S]\setminus I$, we have an instruction $\READ_{I, i}\in \cI_\cG(S)$, which applies the query oracle on the qubits indexed by $I$ and $i$. I.e., given an input $x \in \ZO^n$, the instruction $\READ_{I, i} \in \cU(2^S)$ applies the unitary transformation defined on each basis state by:
\[
\READ_{I, i} \cdot \ket{y_I} \otimes \ket{y_i} \otimes \ket{y_J} = \ket{y_I} \otimes \ket{y_i\oplus x_{y_I}} \otimes \ket{y_J},
\]
where $J = [S] \setminus (I \cup \{i\})$.

Hence if we have a sequence $C_1, \ldots, C_T$ of instructions and an input $x$, we may obtain the state of the memory at time step $t$, on input $x$, by $\ket{\psi_0(x)} = \ket{0}^{\otimes S}$ and $\ket{\psi_{t+1}(x)} = C_{t+1} \ket{\psi_t(x)}$.

We say that a quantum circuit $C = (n, T, S, C_1, \ldots, C_T)$ \textit{computes} or \textit{solves a relation $F \subseteq \ZO^n\times\ZO^m$ with error $\eps$} if $C$ is such that, for every input $x \in \ZO^n$, if we measure the first $m$ qubits of $\ket{\psi_T(x)}$ in the computational basis, we obtain, with probability $\ge 1 - \eps$, a string $z \in \ZO^m$ such that $(x, z) \in F$.

% We let $\QCircuit_\eps(T,S)$ denote the set of Boolean relations $F$ that can be computed by some circuit with error $\eps$, time complexity $O(T)$ and space complexity $O(S)$, over the gate set $\cB$. Likewise, we let $\QRAG_\eps(T,S)$ denote the set of Boolean relations $F$ that can be computed by some circuit with error $\eps$, time complexity $O(T)$ and space complexity $O(S)$, over the gate set $\cQ$. If $\eps$ is omitted, it may be taken to be any suitable constant strictly between $0$ and $1/2$.

\subsection{Quantum Random-Access Machine (QRAM)}
\label{sec:QRAM}

Generally speaking, a quantum circuit is allowed to apply any of the basic operations to any of its qubits. In the definition given above, a quantum random-access gate can specify any permuted subset of the qubits to serve as its inputs. This allows for unusual circuit architectures, which are undesirable.

One may then define a more restricted class of circuits, as follows. We think of the qubits as divided into two parts: work qubits and memory qubits. We have $M$ memory qubits and $W = O(\log M)$ work qubits, for a total space complexity $S = W + M$. We restrict the circuit so that any unitary gate $G \in \cB$, or read instruction, must be applied to work qubits only. And, finally, any random-access gate must be applied in such a way that the addressing qubits ($i$) and the swap qubit ($b$) are always the first $\log M + 1$ work qubits, and the addressed qubits ($x_0, \ldots, x_{M-1}$) are exactly the memory qubits, and are always addressed in the same, fixed order, so one can speak of \textit{the first memory qubit, the second memory qubit}, \emph{etc}. We may then think of a computation as alternating between doing some computation on the work registers, then swapping some qubits between work and memory registers, then doing some more computation on the work registers, and so forth. The final computational-basis measurement is also restricted to measuring a subset of the work qubits.

Under these restrictions, a circuit of time complexity $T$ may be encoded using $O(T \log S)$ bits, whereas in general one might need $\Omega(T S)$ qubits in order to specify how the wires of the circuit connect to the random-access gates.

We will then use the term a \textit{quantum random-access machine algorithm}, or \textit{QRAM algorithm}, for a family of circuits that operate under these restrictions.\footnote{Such a computational model has been referred to by several names in the past. For instance, the term \emph{QRAQM} appears in several publications, starting with \cite{Kuperberg-2005}, and \emph{QAQM} has also been used~\cite{Naya-Plasencia-Schrottenloher-2020}.}
% We let $\QRAM_\eps(T, W, M) \subseteq \QRAG_\eps(T, W + M)$ denote the set of Boolean relations $F$ that can be computed by such a restricted circuit with $O(W)$ work qubits, error $\eps$, time complexity $O(T)$, and $O(M)$ memory qubits. If $\eps$ is omitted, one may take any suitable $\eps \in (0, 1/2)$. 

% It may be seen that any QRAM with $W$ work qubits can be simulated by a QRAM with $\log n + 3$ work qubits, by keeping the original work qubits in memory and doing additional swaps; the time complexity will increase by a constant factor, and the space by an additive $W$; so the definitions do not actually depend on the specific choice of $W$, and $W$ may be omitted, in which case we may take $W = \log n + 3$.

\todo{Some words on the overloaded use of the term "QRAM", and what it means to different people.}
\todo{Mention terms that people have used in the past. For instance, QRAQM appears in a lot of papers, first having been used in ``Another subexponential-time quantum algorithm for the dihedral hidden subgroup problem'', and QAQM has also been used in ``Optimal Merging in Quantum k-xor and k-sum Algorithms''.}

\subsection{Sparse \texorpdfstring{$\QRAM$}{QRAM} algorithms}

In classical algorithms, we may have an algorithm which uses $M$ memory registers, but such that, at any given time, only $m$ out of these $M$ registers are non-zero. In this case we could call such an algorithm \textit{$m$-sparse}. The following definition is the quantum analogue of this.

\begin{definition}
Let $\circuitDesc{C} = (n,T, W, M, C_1, \ldots, C_T)$ be a $\QRAM$ algorithm using time $T$, $W$ work qubits, and $M$ memory qubits. Then, we say that $C$ is $m$-sparse, for some $m \le M$, if at every time-step $t \in \{0, \ldots, T\}$ of the algorithm, the state of the memory qubits is supported on computational basis vectors of Hamming weight $\le m$. I.e., we always have
\[
\ket{\psi_t} \in \lspan \left(  \ket{u}\ket{v} \;\middle|\; u \in \ZO^W, v \in \binom{[M]}{\le m} \right)
\]
In other words, if $\ket{\psi_t}$ is written in the computational basis:
\[
\ket{\psi_t}=\sum_{u \in \ZO^W} \sum_{v \in \ZO^M} \alpha^{(t)}_{u,v} \cdot \underbrace{\ket{u}}_{\text{Work qubits}}\otimes \underbrace{\ket{v}}_{\text{Memory qubits}},
\]
then $\alpha^{(t)}_{u,v} = 0$ whenever $|v| > m$.
\end{definition}

\todo{Make remark here about different choice of basis. If we have an algorithm that is initialized to a non-zero state, we should define sparsity relative to a basis for which this non-zero state is actually all zeros.}

\subsection{Time complexity of simple operations (the constant \texorpdfstring{$\gamma$}{gamma})}\label{sec:Definitions-gamma}

Throughout the paper we will often describe algorithms that use certain simple operations over a logarithmic number of bits. These may include comparison, addition, bitwise XOR, swapping, and others. In a classical random-access machine, all of these operations can be done in $O(1)$ time, as in such machines it is usually considered that every memory position is a register that can hold $O(\log n)$ bits, and such simple operations are taken to be machine instructions.

We do not necessarily wish to make such an assumption for quantum algorithms, since we do not really know what a quantum computer will look like, just yet. So we will broadly postulate the existence of a quantity~$\gamma$, which is an upper-bound on the time complexity of doing such simple operations. We then express our time upper-bounds with $\gamma$ as a parameter. Depending on the precise architecture of the quantum computer, one may think of $\gamma$ as being $O(1)$, or $O(\log n)$. In all our bounds, the simple operations that we will make use of can always be implemented using $O(\log M)$ elementary gates.

\subsection{Controlled unitaries}

Sometimes we will explain how to implement a certain unitary, and we wish to have a version of the same unitary which can be activated or deactivated depending on the state of an additional control bit. We will make free use of the following lemma, which we state without proof.

\begin{lemma}\label{lem:controlled-gates}
If a unitary $U$ can be implemented using $T$ gates from $\cQ$, then the unitary
\[
\ket{b}\ket{x}  \mapsto \begin{cases}
    \ket{b}(U \ket x)  & \text{if } b = 1\\
    \ket{b} \ket x  & \text{if } b = 0
\end{cases}
\]
can be implemented (without error) using $O(T)$ gates from $\cQ$.
\end{lemma}

\section{Compressing sparse algorithms using quantum radix trees}
\label{sec:CompressingUsingRadixTrees}

Let $\circuitDesc{C}=(n, T, W, M, C_1, \ldots, C_T)$ be the circuit of an $m$-sparse $\QRAM$ algorithm computing a relation $F$ with error $\eps$ and let the state of the algorithm at every time-step $t$, when written in the computational basis, be
\begin{equation}
\label{eq:StateEvolutionOriginalAlgorithm}
\ket{\psi_t}=\sum_{u \in \ZO^W} \sum_{v \in \binom{[M]}{\le m}} \alpha^{(t)}_{u,v} \underbrace{\ket{u}}_{\text{Work qubits}}\otimes \underbrace{\ket{v}}_{\text{Memory qubits}}.
\end{equation}
 Using the description of $\circuitDesc{C}$ and the assumption that this algorithm is $m$-sparse we will now construct another $\QRAM$ algorithm $\circuitDesc{C}'$ that uses much less space ($O(m\log M)$ qubits) and computes $F$ with almost same error probability with only $O(\log M \log T)$ factor worsening in the run time.

\paragraph{Main observation.} As the state of the memory qubits in $\ket{\psi_t}$ for any $t$ is only supported on computational basis vectors of Hamming weight at most $m$, one immediate way to improve on the space complexity is to succinctly represent the state of the sparsely used memory qubits. The challenge, however, is that every instruction $C_i$ in $\circuitDesc{C}$ might not have an \emph{easy} analogous implementation in the succinct representation. So we will first present a succinct representation and then show that, for every instruction $C_i$ in the original circuit $\circuitDesc{C}$, there is an analogous instruction or a series of instructions that evolve the state of the succinct representation in the same way as the original state evolves due to the application of $C_i$.

\paragraph{A succinct representation.} Let $v \in \ZO^M$ be a vector with $|v|\le m$ (with $\ket{v}$ being the corresponding quantum state that uses $M$ qubits). Whenever $m$ is significantly smaller than $M$ (i.e., $m\log M < M$) we can instead represent the vector $v$ using the list of indices $\{i\}$ such that $v[i]=1$. Such a representation will use much fewer (qu)bits. Let $S_v$ denote the set of indices $i$ such that $v[i]=1$. We will then devise a quantum state $\ket{S_v}$, that represents the set $S_v$ using a quantum data structure. This representation will be unique, meaning that for every sparse computational-basis state $\ket{v}$ there will be a unique corresponding quantum state $\ket{S_v}$, and $\ket{S_v}$ will use much fewer qubits. Then for every time-step $t$, the quantum state $\ket{\psi_t}$ from Equation~(\ref{eq:StateEvolutionOriginalAlgorithm}) has a corresponding \emph{succinctly represented} quantum state $\ket{\phi_t}$ such that 

\begin{equation}
\label{eq:StateEvolutionSuccinctAlgorithm}
\ket{\phi_t}=\sum_{u \in \ZO^w} \sum_{v \in \binom{[M]} {\le m}} \alpha_{u,v}^{(t)} \ket{u} \otimes \ket{S_v}.    
\end{equation}

By using such a succinct representation, we will be able to simulate the algorithm $\circuitDesc{C}$ with $O(m \log M)$ qubits, with an $O(\gamma \log \frac{T}{\delta})$ additional factor overhead in time and an additional $\delta$ probability of error.

To obtain the desired succinct representation $\ket{S_v}$, we use the quantum radix trees appearing in an algorithm for the subset-sum problem by Bernstein, Jeffery, Lange, and Meurer
\cite{Bernstein-Jeffery-Lange-Meurer-SubsetSum-2013} (see also~\cite{Stacey-PhDThesis-2014}). Several crucial aspects of the implementation were missing or buggy, and required some amount of work to complete and fix. The resulting effort revealed, in particular, that the data-structure is unlikely to be implementable efficiently without error (as it relies on a particular gate which cannot be implemented in an error-free way using the usual basic gates). So we here include all the required details.%\footnote{More generally, Jeffery's idea can be applied to convert any classical \textit{pointer-based} history-independent data structure with reversible updates, into a history-independent quantum data structure.}

\subsection{Radix Tree}

A quantum radix tree is a quantum data structure inspired by the classical radix tree whose definition is as follows. 

\begin{definition}
\label{def:RadixTrees}
A radix tree is a rooted binary tree, where the edges are labeled by non-empty binary strings, and the concatenation of the labels of the edges along any root-to-leaf path results in a string of the same length $\ell$ (independent of the chosen root-to-leaf path). The value $\ell$ is called the \emph{word length} of the tree.

There is a bijective correspondence between radix trees $R$ of word length $\ell$ and subsets $S \subseteq \ZO^\ell$. Given $R$, we may obtain $S$ as follows. Each root-to-leaf path of $R$ gives us an element $x \in S$, so that $x$ is the concatenation of all the edge labels along the path. 

If $R$ corresponds to $S$, we say that $R$ \emph{stores}, or \emph{represents} $S$, and write $R(S)$ for the radix tree representing $S$, i.e., for the inverse map of what was just described (see below).
\end{definition}

An example of a radix tree appears in Figure~\ref{fig:radixTree}.\todo{Figure}

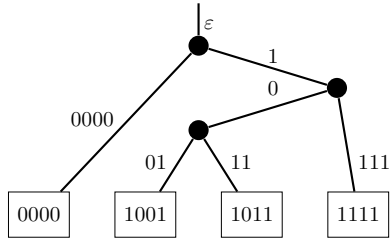
\begin{figure}[h]
    \centering
    \scalebox{0.8}{
\begin{tikzpicture}[scale=0.35]

\draw (-3.5,2.5) node [shape =circle, white, fill=white] (empty){};
\draw (-3,1) node 
  [shape=rectangle, minimum height=0.75cm, minimum width=1cm]
 (Eps){$\eps$};

\draw (-3.5,0) node [shape =circle, black,fill=black] (node1){};

\draw (3,-2) node [shape =circle, black,fill=black] (node21){};

\draw (-3.5,-4) node [shape =circle, black,fill=black] (node31){};

\draw (-11,-8) node 
  [shape=rectangle, minimum height=0.75cm, minimum width=1cm, draw=black]
 (Rectangle1){$0000$};

\draw (-6,-8) node 
  [shape=rectangle, minimum height=0.75cm, minimum width=1cm, draw=black]
 (Rectangle2){$1001$};
\draw (-1,-8) node 
  [shape=rectangle, minimum height=0.75cm, minimum width=1cm, draw=black]
 (Rectangle3){$1011$};
 
\draw (4,-8) node 
  [shape=rectangle, minimum height=0.75cm, minimum width=1cm, draw=black]
 (Rectangle4){$1111$};

\draw [line width=1pt,-] (node1) -- (Rectangle1); 
\draw (-8.5,-3.5) node 
  [shape=rectangle, minimum height=0.75cm, minimum width=1cm]
 (Label0000){$0000$}; 
\draw [line width=1pt,-] (node1) -- (node21);
\draw (0,-0.5) node 
  [shape=rectangle, minimum height=0.75cm, minimum width=1cm]
 (Label1){$1$}; 
\draw [line width=1pt,-] (node21) -- (node31);
\draw (0,-2) node 
  [shape=rectangle, minimum height=0.75cm, minimum width=1cm]
 (Label0){$0$}; 
\draw [line width=1pt,-] (node21) -- (Rectangle4);
\draw (4.75,-5.5) node 
  [shape=rectangle, minimum height=0.75cm, minimum width=1cm]
 (Label111){$111$};
\draw [line width=1pt,-] (node31) -- (Rectangle2);
\draw (-5.5,-5.5) node 
  [shape=rectangle, minimum height=0.75cm, minimum width=1cm]
 (Label01){$01$};
\draw [line width=1pt,-] (node31) -- (Rectangle3);
\draw (-1.5,-5.5) node 
  [shape=rectangle, minimum height=0.75cm, minimum width=1cm]
 (Label11){$11$};

\draw [line width=1pt,-] (empty) -- (node1);
\end{tikzpicture}
}
    \caption{A radix tree storing the set $\{0000, 1001, 1011, 1111\}$} 
    \label{fig:radixTree}
\end{figure}

Given a set $S \subseteq \ZO^\ell$, we obtain $R(S)$ recursively as follows: The empty set corresponds to the tree having only the root and no other nodes. We first find the longest common prefix $p \in \ZO^{\le \ell}$ of $S$. If $|p| > 0$, then we have a single child under the root, with a $p$-labeled edge going into it, which itself serves as the root to $R(S')$, where $S'$ is the set of suffixes (after $p$) of $S$. If $|p| = 0$, then the root will have two children. Let $S=S_0 \cup S_1$, where $S_0$ and $S_1$ are sets of strings starting with $0$ and $1$, respectively, in $S$. The edges to the left and right children will be labeled by $p_0$ and $p_1$, respectively, where $p_0 \in \ZO^{\le \ell}$ is the longest common prefix of $S_0$ and $p_1 \in \ZO^{\le \ell}$ is the longest common prefix of $S_1$. The left child serves as a root to $R(S'_0)$, where $S'_0$ is the set of suffixes (after $p_0$) of $S_0$. Analogously, the right child serves as a root to $R(S'_1)$, where $S'_1$ is the set of suffixes (after $p_1$) of $S_1$. 

\paragraph{Basic operations on radix trees} The allowed basic operations on a radix tree are insertion and removal of an element. Classically, an attempt at inserting an element already in $S$ will result in the identity operation. Quantumly, we will instead allow for \textit{toggling} an element in/out of $S$.

\paragraph{Representing a radix tree in memory} We now consider how one might represent a radix tree in memory. For this purpose, suppose we wish to represent a radix tree $R(S)$ for some set $S \subseteq \ZO^\ell$ of size $|S| \le m$. Let us assume without loss of generality that $m$ is a power of $2$, and suppose we have at our disposal an array of $2m$ memory blocks.

Each memory block may be used to store a node of the radix tree. If we have a node in the tree, the contents of its corresponding memory block will represent a tuple $(z, p_1, p_2, p_3)$. The value $z \in \ZO^{\le \ell}$ stores the label in the edge from the node's parent, the values $p_1,p_2,p_3 \in \{0, 1, \ldots, 2 m\}$ are pointers to the (block storing the) parent, left child, and right  child, respectively, or $0$ if such an edge is absent. %We represent $z$ by a binary string of length $\lceil \log \ell \rceil + \ell$, specifying the length of $z$ and then the contents of $z$ itself, concatenated with $0$s. Each pointer is represented using $\lceil \log n \rceil$ bits.

It follows that each memory block is $O(\ell + \log m)$ bits long. In this way, we will represent $R(S)$ by a binary string of length $O(m (\ell + \log m))$. The root node is stored in the first block, empty blocks will be set to $0$, and the only thing that needs to be specified is the \textit{memory layout}, namely, in which block does each node get stored. For this purpose, let $\tau:R(S) \to [2m]$ be an injective function, mapping the nodes of $R(S)$ to the $[2m]$ memory blocks, so that $\tau(\text{root}) = 1$. For any $S \subseteq \ZO^\ell$ of size $|S|\leq m$, we then let
\[
R_\tau(S) \in \ZO^{O(m (\ell + \log m))}
\]
denote the binary string obtained by encoding $R(S)$ as just described.

\paragraph{BJLM's quantum radix tree} We see now that although there is a unique radix tree $R(S)$ for each $S$, there is no obvious way of making sure that the representation of $R(S)$ in memory is also unique. However, this bijective correspondence between $S$ and its memory representation is a requirement for quantum algorithms to use interference. The idea of Bernstein et al \cite{Bernstein-Jeffery-Lange-Meurer-SubsetSum-2013}, then, is to represent $S$ using a superposition of \textit{all possible layouts}. I.e., $S$ is to be uniquely represented by the (properly normalized) quantum state:
\[
\sum_\tau \ket{R_\tau(S)}.
\]
The trick, then, is to ensure that this representation can be efficiently queried and updated. In their discussion of how this might be done, the BJLM paper \cite{Bernstein-Jeffery-Lange-Meurer-SubsetSum-2013} presents the broad idea but does not work out the details, whereas Jeffery's thesis \cite{Stacey-PhDThesis-2014} glosses over several details and includes numerous bugs and omissions. To make their idea work, we make use of an additional data structure.

\subsection{Prefix-Sum Tree}
\label{sec:PrefixSumTree}

In our implementation of the Quantum Prefix Tree, we will need to keep track of which blocks are empty and which are being used by a node. For this purpose, we will use a data-structure that is famously used to (near-optimally) solve  the dynamic prefix-sum problem.

\begin{definition}
\label{def:PrefixSumTree}
A \textit{prefix-sum tree} is a complete rooted binary tree. Each leaf node is labelled by a value in $\ZO$, and each internal node is labelled by the number of $1$-valued leaf nodes descending from it.

Let $F \subseteq [\ell]$ for $\ell$ a power of $2$. We use $P(F)$ to denote the prefix-sum tree where the $i^{th}$ leaf node of the tree is labelled by $1$ iff $i \in F$.
\end{definition}

% A prefix-sum tree of $\ell$ leaves will have a total of $2 \ell-1$ nodes. An example of such a tree appears in Figure~\ref{}.

A prefix-sum tree $P(F)$ will be represented in memory by an array of $\ell - 1$ blocks of memory, holding the labels of the inner nodes of $P(F)$, followed by $\ell$ bits, holding the labels of the leaf nodes. The blocks appear in the same order as a breadth-first traversal of $P(F)$. Consequently, for every $F \in \ZO^{\ell}$ there is corresponding binary string of length $(\ell-1)\log \ell + \ell$ that uniquely describes $P(F)$.

We will overload notation, and use $P(F)$ to denote this binary string of length $(\ell-1)\log \ell + \ell$.

\paragraph{Allocating and deallocating.}
The idea now is to use the prefix tree as an memory allocator. We have $2 m$ blocks of memory, and the set $F$ will keep track of which blocks of memory are unused, or ``free''. 

We would then like to have an operation that allocates one of the free blocks. To implement Bernstein et al's idea, the choice of which block to allocate is made in superposition over all possible free blocks. I.e., we would like to implement the following map $\unitaryAlloc$ and also its inverse, $\unitaryFree$.
\begin{equation}
\label{eq:UnitaryAlloc}
    \unitaryAlloc: 
    \ket{P(F)}\ket{0}  \ket{0}\rightarrow \frac{1}{\sqrt{|F|}}\sum_{i \in  F} \ket{P(F\setminus\{i\})}\ket{i} \ket{0},
\end{equation}
The second and third registers have $O(\log m)$ bits. We do not care for what the map does when these registers are non-zero, or when $F = \varnothing$. We will guarantee that this is never the case.
%Although this should never happen if one is careful.
%\todo{FS: I don't really understand this yet from the point of unitarity. What if we apply the map twice on an F with a single 1?}

Note that each internal node of the prefix tree stores the number of elements of $F$ that are descendants to that node. In particular, the root stores $|F|$. In order to implement $\unitaryAlloc$, we then start by constructing the state
\begin{equation}
    \frac{1}{\sqrt{|F|}} \sum_{j=1}^{|F|}\ket{j}.\label{eq:rand-number-F}
\end{equation}
While this might appear to be simple, it actually requires us to use a gate
\begin{equation}
\label{eq:UnitarySupperpose}
\unitarySuperpose: \ket{k}\ket{0} \mapsto \frac{1}{\sqrt{k}} \sum_{j=1}^{k} \ket{k} \ket{j}.
\end{equation}

This is much like choosing a random number between $1$ and a given number $k$ on a classical computer. Classically, such an operation cannot be done exactly if all we have at our disposition are bitwise operations (since all achievable probabilities are then dyadic rationals). Quantumly, it is impossible to implement $\unitarySuperpose$ efficiently without error by using only the usual set of basic gates.%\todo{write footnote}\footnote{More precisely, if one uses only the gates in $\cB$, then \dots}

\label{claim:unitary-superpose} 
So the reader should take note: it is precisely this gate which adds error to BJLM's procedure. This gate can be implemented up to distance $\eps$ using $O(\log\frac{m}{\eps})$ basic gates, where $m$ is the maximum value that $k$ can take. I.e., using so many gates we can implement a unitary $U$ such that the spectral norm $\|U - \unitarySuperpose\| \le \eps$.\footnote{This is done by using Hadamard gates to get a superposition between $1$ and the smallest power of $2$ which is greater than $\frac{m}{\eps}$, and then breaking this range into $m$ equal intervals plus a remainder of size $< m$. The \textit{remainder subspace} will have squared amplitude $\le \eps$.} We will need to choose $\eps \approx \frac{1}{T}$, which is the inverse of the number of times such a gate will be used throughout our algorithm.

Once we have prepared state (\ref{eq:rand-number-F}), we may then use binary search, going down through the prefix tree to find out which location $i$ corresponds to the $j^{th}$ non-zero element of $F$. Using $i$, as we go up we can remove the corresponding child from $P(F)$, in $O(\gamma \cdot \log m)$ time, while updating the various labels on the corresponding root-to-leaf path. This requires the use of $O(\log m)$ work bits, which are $\ket{0}$ at the start and end of the operation.
During this process, the register holding $j$ is also reset to $\ket{0}$, by subtracting the element counts we encounter during the deletion process from this register.
The inverse procedure $\unitaryFree$ is implemented in a similar way.

\subsection{Quantum Radix Tree}

We may now define the quantum radix tree.

\begin{definition}[Quantum Radix Tree]
\label{def:QuantumRadixTree}
Let $\ell$ and $m$ be powers of $2$, $S \subseteq \ZO^\ell$ be a set of size $s=|S| \le m$, and let $R(S)$ be the classical radix tree storing $S$.
Then, the \textit{quantum radix tree} corresponding to $S$, denoted $\ket{R_Q(S)}$ (or $\ket{R_Q^{\ell,m}(S)}$ when $\ell$ and $m$ are to be explicit), is the state
\[
\ket{R_Q(S)} = \frac{1}{\sqrt{N_S}} \cdot \sum_{\tau} \ket{R_\tau(S)} \ket{P(F_\tau)},
\]
where $\tau$ ranges over all injective functions $\tau:R(S)\to [2m]$ with $\tau(\text{root}) = 1$, of which there are $N_S = \frac{(2m-1)!}{(2m-|R(S)|)!}$ many, and $F_\tau = [2m]\setminus \tau(R(S))$ is the complement of the image of $\tau$.
\end{definition}

\paragraph{Basic operations on quantum radix trees}
The basic allowed operations on a quantum radix trees are look-up and toggle, where the toggle operation is analogous to insertion and deletion in classical radix tree. Additionally, we also define a swap operation which will be used to simulate a RAG gate.

\begin{lemma}\label{thm:quantumRadixTree}
Let $\ket{\QRT{S}} = \ket{R_Q^{\ell,m}(S)}$ denote a quantum radix tree storing a set $S \subseteq \ZO^{\ell}$ of size at most $m$. We then define the following data structure operations.
\begin{enumerate}
    \item \texttt{Lookup.} Given an element $e \in \ZO^\ell$, we may check if $e \in S$, so for each $b \in \ZO$, we have the map
    \begin{equation*}
        \ket{e}\ket{\QRT{S}}\ket{b} \mapsto \ket{e}\ket{\QRT{S}}\ket{b \oplus (e \in S)}.
    \end{equation*}
    \item \texttt{Toggle.} Given $e \in \ZO^\ell$, we may add $e$ to $S$ if $S$ does not contain $e$, or otherwise remove $e$ from $S$. Formally,
    \begin{equation*}
        \ket{e} \ket{\QRT{S}} \mapsto
        \begin{cases}
            \ket{e} \ket{\QRT{S\cup \{e\}}}, & \text{if } e \notin S, \\
            \ket{e} \ket{\QRT{S\setminus \{e\}}}, & \text{if } e \in S. 
        \end{cases}
    \end{equation*}
    \item \texttt{Swap.} Given an element $e \in \ZO^{\ell}$, $b \in \ZO$ and a quantum radix tree storing a set $S$, we would like \texttt{swap} to be the following map,
    \begin{equation*}
        \ket{e}\ket{\QRT{S}}\ket{b} \mapsto
        \begin{cases}
            \ket{e}\ket{\QRT{S \cup \{e\}}}\ket{0}, & \text{ if } e \notin S \text{ and } b=1, \\
            \ket{e}\ket{\QRT{S \setminus \{e\}}}\ket{1}, & \text{ if } e \in S \text{ and } b=0, \\
            \ket{e}\ket{\QRT{S}}\ket{b}, & \text{otherwise.}
        \end{cases}
    \end{equation*}
\end{enumerate}
These operations can be implemented in worst case $O(\gamma \cdot \log m)$ time and will be error-free if we are allowed to use an error-free gate for $\unitarySuperpose$ (defined in Equation~\ref{eq:UnitarySupperpose}), along with other gates from set $\mathcal{Q}$.
\end{lemma}

\begin{proof}
Let $\ket{b},\ket{e}$ denote the quantum states storing the elements $b \in \ZO$ and $e \in \ZO^\ell$, respectively. The data structure operations such as \texttt{lookup}, \texttt{toggle} and \texttt{swap} can be implemented reversibly in $O(\gamma \cdot \log m)$ time in the following way.

\paragraph{\texttt{Lookup}} We wish to implement the following reversible map $\unitaryLook$,
    \begin{equation}
    \label{eq:UnitaryLookQRT}
        \unitaryLook:\ket{e}\ket{\QRT{S}}\ket{b} \mapsto \ket{e}\ket{\QRT{S}}\ket{b \oplus (e \in S)}.
    \end{equation}
We do it as follows. First note that,
by Definition~\ref{def:QuantumRadixTree},
    \begin{equation*}
        \ket{\QRT{S}}=\frac{1}{\sqrt{N_S}} \sum_{\tau} \ket{R_{\tau}(S)}\ket{P(F_{\tau})}.
    \end{equation*}
We will traverse $R_{\tau}(S)$ with the help of some auxiliary variables. Starting at the root node, we find the edge labeled with a prefix of $e$. If no such label is found then $e$ is not present in $R_{\tau}(S)$. Otherwise, we traverse to the child reached by following the edge labeled by a prefix of $e$. Let us denote the label by $L$. If the child is a leaf node then terminate the process, stating that $e$ is present in $R_{\tau}(S)$, else, recurse the process on $e'$ and the tree rooted at that child node. Here $e'$ is the binary string after removing $L$ from $e$. When at some point we have determined whether $e \in S$ or not, we flip the bit $b$, or not. Eventually, we may conclude that $e \not \in S$ before traversing the entire tree, at which point we skip the remaining logic for traversing $R_{\tau}(S)$ downwards (by using a control qubit). After we have traversed $R_\tau(S)$ downwards and determined whether $e \in S$, we need to undo our traversal, which we do by following the $p_1$ pointers (to the parent nodes) until the root is again reached, and the auxiliary variables are again set to $0$.

Each comparison with the edge labels, at each traversed node, takes $O(\gamma)$ time. Hence, the entire procedure takes $O(\gamma \cdot \log m)$ time. % (i.e., it can be implemented reversibly using $O(\gamma \cdot \log m)$ gates from the gate set $\cQ=\{\CNOT, H, S, T, \RAG \}$). We call this unitary $\unitaryLook$. By linearity, we argue that the application $\unitaryLook$ on the state $\ket{\QRT{S}}$ gives us the desired result as mentioned in Equation~\ref{eq:UnitaryLookQRT}.

\paragraph{\texttt{Toggle}} Let $\unitaryToggle$ denote the following map,
    \begin{equation}
        \unitaryToggle: \ket{e} \ket{\QRT{S}} \rightarrow
        \begin{cases}
            \ket{e} \ket{\QRT{S\setminus \{e\}}}, & \text{if } e \in S, \\
            \ket{e} \ket{\QRT{S\cup \{e\}}}, & \text{if } e \notin S 
        \end{cases}
    \end{equation}
    
The \texttt{toggle} operation primarily consists of two main parts: The memory allocation or de-allocation, followed by insertion or deletion, respectively. 

We again traverse $R_{\tau}(S)$ with the help of some auxiliary variables.
We start with the root node of $R_{\tau}(S)$, and traverse the tree downwards until we know, as above, whether $e \in S$ or not. If $e \notin S$, we will know where we need to insert nodes into $R_\tau(S)$, in order to transform it into $R_\tau(S \cup \{e\})$. Below, we will explain in detail how such an insertion must proceed. It turns out that we may need to insert either one node, or two, but never more. We may use the work qubits to compute the contents of the memory blocks that will hold this new node (or new nodes). These contents are obtained by XORing the appropriate bits of $e$ and the appropriate parent/child pointers of the nodes we are currently traversing in the tree.

We may then use the $\unitaryAlloc$ gate (once or twice) to obtain the indices of the blocks that will hold the new node(s). We then use RAG gates to swap in the contents of these blocks into memory. A fundamental and crucial detail must now be observed: the index of the memory blocks into where we inserted the new nodes is now left as part of the work qubits. This cannot be and must be dealt with, because every work bit must be again set to zero at the end of the procedure. However, a copy of this index now appears as the child pointer ($p_2$ or $p_3$) of the parents of the nodes we just created, and these pointers can thus be used to zero out the index. It is then possible to traverse the tree upwards in order to undo the various changes we did to the auxiliary variables.

If $e \in S$, on the other hand, we then do the inverse procedure. We will then know which nodes need to be removed from $R_\tau(S)$ (it will be either one or two nodes). By construction, these nodes will belong to blocks not in $F_\tau$. We begin by setting these blocks to zero by swapping the blocks into the workspace (using the RAG gate), XORing the appropriate bits of $e$ and the appropriate child/parent pointers so the blocks are now zero, and swapping them back. These blocks will then be set to zero, and we are left with a state akin to the right-hand side of (\ref{eq:UnitaryAlloc}). We then use the $\unitaryFree$ gate to \textit{free} the blocks, i.e., add their indices to $F_\tau$ once again. At this point we can traverse the tree upwards once more, in order to reset the auxiliary variables to zero, as required.

We now give further detail on how one must update $R_\tau(S)$ in order to insert a new element $e$ into $S$. We must create a node $N \coloneqq (z, p_1, p_2, p_3)$ corresponding to the element $e$ stored at the memory location assigned by $\unitaryAlloc$ procedure. Let us denote the address by $k$. Start with the root node of $R_{\tau}(S)$. If $e$ has no common prefix with any of the labels of the root's outgoing edges, which can only happen if the root has one child, then set $z$ to $e$, $p_1$ pointing to the root node, and, $p_2$ and $p_3$ set to $0$. Moreover, set the value of the root's $p_2$ pointer to $k$ if node $N$ ends up as the left child to the root, else set root's $p_3$ pointer to $k$. In the case when $e$ has a common prefix with one of the labels of the root's outgoing edges, let us denote the label by $L$ and the child node by $C$, then further two scenarios arise: Either label $L$ is completely contained in $e$, which if is the case then we traverse the tree down and run the insertion procedure recursively on $e'$ (which is $e$ after removing the prefix $L$) with the new root set $C$. In the case where label $L$ is not completely contained in $e$, we create an internal node $N'$ with its $z$ variable set to the longest common prefix of $e$ and $L$ (which we denote by $L'$), $p_1$ pointing to root, $p_2$ pointing to $C$ and $p_3$ pointing to $N$ (or vice versa depending on whether node $N$ gets to be the right or the left child). We run the $\unitaryAlloc$ procedure again to get a memory location to store $N'$. Having done that, we now change the $z$ value of node $C$ to be the prefix of $L$ after $L'$, and the $p_1$ value of node $C$ to be the memory location of $N'$. Additionally, we also set $z$ of node $N$ to be $e'$, the suffix of $e$ after $L'$, and we let $p_1, p_2, p_3$ to be, respectively, a pointer to $N'$, $0$ and $0$.\todo{image}

Each step in the traversal takes time $O(\gamma)$, for a total time of $O(\gamma \cdot \log m)$.

The procedure to update $R_{\tau}(S)$ in order to delete an element $e$ from $S$ is analogous to the insertion procedure mentioned above, which also can be implemented in $O(\gamma \cdot \log m)$ time.

\paragraph{\texttt{Swap}} Let $\unitarySwap$ denote the following map,
\begin{equation*}
        \unitarySwap:\ket{e}\ket{\QRT{S}}\ket{b} \mapsto
        \begin{cases}
            \ket{e}\ket{\QRT{S \cup \{e\}}}\ket{0}, & \text{ if } e \notin S \text{ and } b=1, \\
            \ket{e}\ket{\QRT{S \setminus \{e\}}}\ket{1}, & \text{ if } e \in S \text{ and } b=0, \\
            \ket{e}\ket{\QRT{S}}\ket{b}, & \text{otherwise.}
        \end{cases}
    \end{equation*}
To implement $\unitarySwap$, we first run the $\unitaryLook$ on the registers $\ket{e}$, $\ket{\QRT{S}}$ and $\ket{b}$. Conditional on the value of register $\ket{b}$ (i.e., when $b=1$), we run $\unitaryToggle$ on the rest of the registers. We then run $\unitaryLook$ again to attain the desired state. To summarize, the unitary $\unitarySwap=\unitaryLook \cdot \controlledToggle \cdot \unitaryLook$, where $\controlledToggle$ is controlled version of $\unitaryToggle$ (as per Lemma \ref{lem:controlled-gates}). Thus, the \texttt{swap} procedure takes a total time of $O(\gamma \cdot \log m)$.
\end{proof}

An error-less, efficient implementation of the unitary $\unitarySuperpose$ is impossible by using only the usual sets of basic gates.\todo{Substantiate this claim.} Furthermore, it is unreasonable to expect to have an error-free $\unitarySuperpose$ at our disposal. However, as we explained in page \pageref{claim:unitary-superpose}, there is a procedure to implement $\unitarySuperpose$ using gates from the gate set $\mathcal{B}=\{ \CNOT, H, S, T\}$ up to spectral distance $\eps$, using only $O(\log \frac{m}{\epsilon})$ gates.

\begin{corollary}\label{thm:OperationsRadixTreeError}
Let $\ket{\QRT{S}} = \ket{R_Q^{\ell,m}(S)}$ denote a quantum radix tree storing a set $S \subseteq \ZO^{\ell}$ of size at most $m$. The data structure operations \texttt{look-up}, \texttt{toggle} and \texttt{swap}, as defined in the statement of Lemma~\ref{thm:quantumRadixTree} can be implemented in $O(\gamma \cdot \log \frac{m}{\epsilon})$ time and $\epsilon$ probability of error using gates from the gate set $\mathcal{Q}$. Here $\gamma$ is the number of gates required from set $\mathcal{Q}$ to do various basic operations on a logarithmic number of qubits.
\end{corollary}   

\subsection{The simulation}
\label{sec:Simulation}

Recall from Section \ref{sec:Definitions-gamma} that we take $\gamma$ to be the number of gates required to do various basic operations on a logarithmic number of qubits. In our use below, it never exceeds $O(\log M)$.

\begin{thm}
Let $T$, $W$, $m < M = 2^\ell$ be natural numbers, with $M$ and $m$ both powers of $2$, and let $\eps \in [0, 1/2)$. Suppose we are given an $m$-sparse QRAM algorithm using time $T$, $W$ work qubits and $M$ memory qubits, that computes a Boolean relation $F$ with error $\eps$.

Then we can construct a QRAM algorithm which computes $F$ with error $\eps' > \eps$, and runs in time $O(T \cdot \log(\frac{ T}{\eps' - \eps}) \cdot \gamma)$, using $W + O(\log M)$ work qubits and $O(m \log M)$ memory qubits.
\end{thm}

\begin{proof}
Let $\circuitDesc{C}=(n, T, W, M, C_1, \ldots, C_T)$ be the circuit of the given $m$-sparse $\QRAM$ algorithm computing a relation $F$ with error $\eps$ and, let the state of the algorithm at every time-step $t$, when written in the computational basis be
\begin{equation}
\ket{\psi_t}=\sum_{u \in \ZO^w} \sum_{v \in \binom{[M]}{\le m}} \alpha^{(t)}_{u,v} \cdot \underbrace{\ket{u}}_{\text{W qubits}}\otimes \underbrace{\ket{v}}_{\text{M qubits}}
\end{equation}
where the set $\binom{[M]}{\le m}$ denotes all vectors $v \in \ZO^M$ such that $|v| \le m$. Using the description of $\circuitDesc{C}$ and the fact that this algorithm is $m$-sparse we will now construct another $\QRAM$ algorithm $\circuitDesc{C}'$ with the promised bounds. 
The algorithm $C'$ will have $w' = W + O(\log M)$ work bits, and $O(m \log M)$ memory bits. The memory is to be interpreted as an instance $\ket{R_Q(S)}$ of the quantum radix tree described above. Then $\ket{v}$ will be represented by the quantum radix tree $\ket{R_Q(S_v)}$, where $S_v = \{ i \in [M] \mid v_i = 1 \}$ is the set of positions where $v_i = 1$, so that each position $i \in [M]$ is encoded using a binary string of length $\ell$.

The simulation is now simple to describe. First, the quantum radix tree is initialized. Then, each non-RAG instruction $C_i \in \circuitDesc{C}$ operating on the work qubits of $\circuitDesc{C}$ is applied in the same way in $\circuitDesc{C}'$ to same qubits among the first $W$ qubits of $\circuitDesc{C}'$. Each RAG instruction, on the other hand, is replaced with the $\unitarySwap$ operation, applied to the the quantum radix tree. The extra work qubits of $\circuitDesc{C}'$ are used as anciliary for these operations, and we note that they are always returned to zero.

If we assume that the $\unitarySwap$ operation can be implemented without error, we then have a linear-space isomorphism between the two algorithms' memory space, which maps the state $\ket{\psi_t}$ of $\circuitDesc{C}$ at each time step $t$ to the state $\ket{\phi_t}$ of $\circuitDesc{C}'$ after $t$ simulated steps:
\[
\ket{\phi_t} = \sum_{u,v} \alpha^{(t)}_{u,v} \cdot \underbrace{\ket{u}}_{W} \otimes \underbrace{\ket{0}}_{O(\log M)} \otimes \underbrace{\ket{R_Q(S_v)}}_{O(m \log M)}.
\]
Thus, if $\unitarySwap$ could be implemented without error, we could have simulated $\circuitDesc{C}$ without additional error.
Otherwise, as per Corollary~\ref{thm:OperationsRadixTreeError}, we may implement the $\unitarySwap$ unitary with an error parameter $\Omega(\frac{\eps' - \eps}{T})$, resulting in a total increase in error of $\eps' - \eps$, and an additional time cost of $O(T \log\frac{T}{\eps' - \eps})$.
\end{proof}

\section{Simplifications of previous work}\label{sec:Simplifications}

It is possible to use our main theorem to simplify the presentation of the following three results: Ambainis' Quantum Walk algorithm for solving the $k$-$\elementDistinctness$ problem~\cite{Ambainis-ElementDistinctness-2004}, Aaronson et al's Quantum algorithms for the Closest Pair problem ($\closestPair$), and the authors' previous paper on Fine-Grained Complexity via Quantum Walks~\cite{buhrman-3SUMhard-2021}.
 
All these results use quantum walk together with complicated, space-efficient, history-independent data structures. As we will see, it is possible to replace these complicated data structures with simple variants of the prefix-sum tree (Section \ref{sec:PrefixSumTree}), where the memory use is sparse, and then invoke the main theorem of our paper.

\subsection{Ambainis' Walk Algorithm for Element Distinctness}\label{sec:ElementDistinctness}
Ambainis' description and analysis of his data structure is complicated, and roughly 6 pages long, whereas a presentation of his results with a simple data structure and an appeal to our theorem requires less than 2 pages, as we will now see. Also, the presentation of the algorithm is considerably muddled by the various difficulties and requirements pertaining to the more complicated data structure. In a presentation of his results that would then appeal to Theorem \ref{thm:main}, we have a very clear separation of concerns.

Ambainis' algorithm is a $\tO(n^{\frac{k}{k+1}})$-time solution to the following problem:

\begin{definition}[$k$-Element Distinctness]
Given a list $L$ of $n$ integers in $\Sigma$ are there $k$ elements $x_{i_1}, \ldots, x_{i_k} \in L$ such that $x_{i_1}=\dots=x_{i_k}$.
\end{definition}

Ambainis' algorithm for $k$-Element Distinctness \cite{Ambainis-ElementDistinctness-2004} is quantum walk algorithm on a Johnson graph $J(n,r)$ with $r=n^{k/k+1}$ and runs in $\tO(n^{k/k+1})$ time. The crucial ingredient in making the algorithm time efficient is the construction of data-structure which can store a set $S \subseteq [n] \times \Sigma$ of elements of size $r$, under efficient insertions and removals, so that one may efficiently query at any given time whether there exist $k$ elements $(i_1, x_1), \ldots, (i_k, x_k)$ in $S$ with distinct indices $i_1, \ldots, i_k$ but equal labels $x_1 = \dots = x_k$. Ambainis makes use of skip-lists and hash tables, ensuring that all operations run in $O(\log^4(n+|\Sigma|))$ time. However, if one does not care about space-efficiency, there is a much simpler data structure that serves the same purpose. The following definition is illustrated in Figure~\ref{fig:AugPrefixTree}.

\begin{definition}\label{def:AugmentedPrefixSum}
Let $S \subseteq [n] \times \Sigma$, with $|S| = r$ and $|\Sigma| = n^{O(1)}$ a power of $2$, and such that every $i \in [n]$ appears in at most one pair $(i,x) \in S$.
The \textit{$k$-element-distinctness tree} that represents $S$, denoted $\kedT(S)$, is a complete rooted binary tree with $|\Sigma|$ leaves. Each leaf node $x \in \Sigma$ is labeled by a bit vector $\bitVector_x \in \ZO^n$ and a number $\countSol_x \in \{0, \ldots, n\}$, so that $\bitVector_x[i] = 1$ iff $(i,x) \in S$, and the $\countSol_x$ is the Hamming weight of $\bitVector_x$. Each internal node $w$ is labeled by a bit $\flag_w \in \ZO$ which indicates whether there exists a leaf $x$, descendent of $w$, with $\countSol_x \ge k$.
\end{definition}

\paragraph{Memory Representation} A $k$-element-distinctness tree is represented in the memory by an array of $|\Sigma| - 1$ bits of memory holding the flags of the internal nodes, followed by $|\Sigma|$ blocks of $n + \lceil \log n \rceil$ bits of memory each, holding the labels of the leaf nodes. The blocks appear in the same order as a breadth-first traversal of $\kedT(S)$. Consequently, for every $S \subseteq [n] \times \Sigma$ there is a corresponding binary string of length $|\Sigma|-1 + (n + \lceil \log n \rceil)|\Sigma|$ that uniquely encodes $\kedT(S)$. Crucially, if $|S| = r$, then at most $O(r (\log \Sigma + \log n))$ of these bits are $1$. So for $|\Sigma| = \poly(n)$, the encoding is $\tO(r)$-sparse.

\begin{figure}[h]
    \centering
    \scalebox{0.8}{
\begin{tikzpicture}[scale=0.5]
\draw (0,0) node 
  [shape=circle, minimum height=1.5cm, draw=black]
 (Level1){\texttt{flag}};
 
\draw (0,-4.5) node 
  [shape=circle, minimum height=1.5cm, draw=black]
 (Level21){\texttt{flag}};
\draw (8,-4.5) node 
  [shape=circle, minimum height=1.5cm, draw=black]
 (Level22){\texttt{flag}};
\draw (-8,-4.5) node 
  [shape=circle, minimum height=1.5cm, draw=black]
 (Level23){\texttt{flag}};

\draw (2,-9) node 
  [shape=circle, minimum height=1.5cm, draw=black]
 (Level31){\texttt{flag}};
\draw (-2,-9) node 
  [shape=circle, minimum height=1.5cm, draw=black]
 (Level32){\texttt{flag}};
\draw (6,-9) node 
  [shape=circle, minimum height=1.5cm, draw=black]
 (Level33){\texttt{flag}};
\draw (-6,-9) node 
  [shape=circle, minimum height=1.5cm, draw=black]
 (Level34){\texttt{flag}};
\draw (10,-9) node 
  [shape=circle, minimum height=1.5cm, draw=black]
 (Level35){\texttt{flag}};
\draw (-10,-9) node 
  [shape=circle, minimum height=1.5cm, draw=black]
 (Level36){\texttt{flag}};
 
\draw (0,-12.5) node 
  [shape=rectangle, minimum height=0.75cm, minimum width=2cm, draw=black]
 (BitVector0){$0$};
\draw (4,-12.5) node 
  [shape=rectangle, minimum height=0.75cm, minimum width=2cm, draw=black]
 (BitVector1){$1$};
\draw (8,-12.5) node 
  [shape=rectangle, minimum height=0.75cm, minimum width=2cm, draw=black]
 (BitVectorDash){...};
\draw (12,-12.5) node 
  [shape=rectangle, minimum height=0.75cm, minimum width=2cm, draw=black]
 (BitVectorNcube){$n^3$};
\draw (-4,-12.5) node 
  [shape=rectangle, minimum height=0.75cm, minimum width=2cm, draw=black]
 (BitVectorMinus1){$-1$};
\draw (-8,-12.5) node 
  [shape=rectangle, minimum height=0.75cm, minimum width=2cm, draw=black]
 (BitVectorMinusDash){...};
\draw (-12,-12.5) node 
  [shape=rectangle, minimum height=0.75cm, minimum width=2cm, draw=black]
 (BitVectorMinusNcube){$-n^3$};

\draw [line width=1pt,->] (Level35) -- (BitVectorNcube); 
\draw [line width=1pt,->] (Level36) -- (BitVectorMinusNcube);
\draw [line width=1pt,->] (Level35) -- (BitVectorDash); 
\draw [line width=1pt,->] (Level36) -- (BitVectorMinusDash);
\draw [line width=1pt,->] (Level33) -- (BitVectorDash); 
\draw [line width=1pt,->] (Level34) -- (BitVectorMinusDash);
\draw [line width=1pt,->] (-6,-10.45) -- (BitVectorMinusDash);
\draw [line width=1pt,->] (6,-10.45) -- (BitVectorDash);
\draw [line width=1pt,->] (Level32) -- (BitVectorMinus1);
\draw [line width=1pt,->] (Level32) -- (BitVector0);
\draw [line width=1pt,->] (Level31) -- (BitVectorDash);
\draw [line width=1pt,->] (Level31) -- (BitVector1);

\draw [line width=1pt,->] (Level21) -- (Level31);
\draw [line width=1pt,->] (Level21) -- (Level32);
\draw [line width=1pt,->] (Level23) -- (Level34);
\draw [line width=1pt,->] (Level23) -- (Level36);
\draw [line width=1pt,->] (Level22) -- (Level33);
\draw [line width=1pt,->] (Level22) -- (Level35);
\draw [line width=1pt,dotted, ->] (1,-1) -- (2.5,-2.5);
\draw [line width=1pt,dotted, ->] (-1,-1) -- (-2.5,-2.5);

\draw [line width=1pt, -] (2,-13.25) -- (4.5,-13.75);
\draw [line width=1pt, -] (-2,-13.25) -- (-4.5,-13.75);

\draw (-3,-14.5) node 
 [shape=rectangle, minimum height=0.75cm,  minimum width=1.5cm, draw=black]
 (BitVectorDescription2){$\countSol_0$};
\draw (1.5,-14.5) node 
  [shape=rectangle, minimum height=0.75cm,  minimum width=1cm, draw=black]
 (BitVectorDescription31){$...$};
\draw (3.5,-14.5) node 
  [shape=rectangle, minimum height=0.75cm, minimum width=1cm, draw=black]
 (BitVectorDescription32){$0/1$};
\draw (3.5,-16) node 
  [shape=rectangle, minimum height=0.75cm, minimum width=1cm]
 (BitVectorDescription42){$n$}; 
 
\draw (-0.5,-14.5) node 
  [shape=rectangle, minimum height=0.75cm, minimum width=1cm, draw=black]
 (BitVectorDescription33){$0/1$};
\draw (-0.5,-16) node 
  [shape=rectangle, minimum height=0.75cm, minimum width=1cm]
 (BitVectorDescription43){$1$};
\end{tikzpicture}
}
    \caption{Data structure for the $k$-Element Distinctness problem.} 
    \label{fig:AugPrefixTree}
\end{figure}
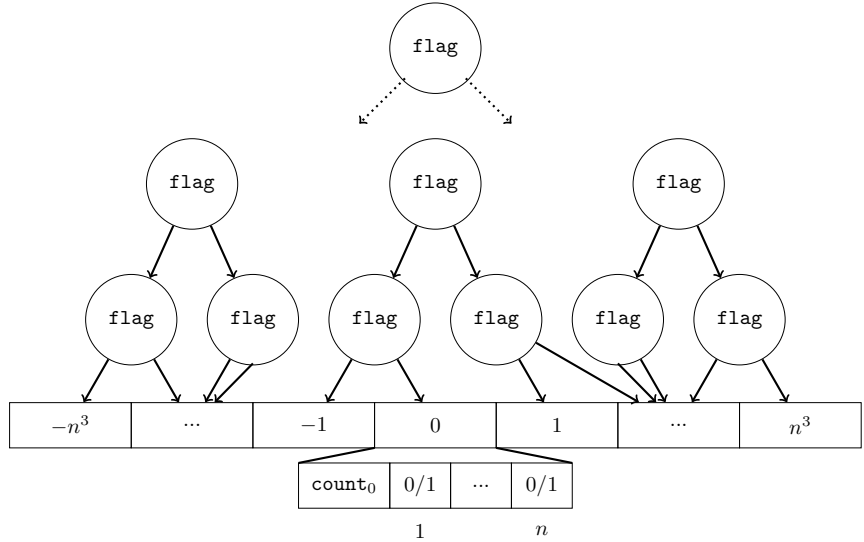

\paragraph{Implemention of data structure operations} It is clear from the definition of $k$-element-distinctness tree and its memory representation that a tree $\kedT(S)$ represents a set $S \subseteq [n]\times \Sigma$ in a history-independent way. We will now argue that all the required data structure operations take $O(\log n)$ time in the worst case. Let $(i, x)$ denote an element in $[n]\times \Sigma$. 
\begin{itemize}
    \item \textbf{Insertion} To insert $(i,x)$ in the tree, first increase the value of the $\countSol$ variable of the leaf $x$, and set $\bitVector_x[i]=1$. Then, if $\countSol \ge k$, set $\flag_w = 1$ for all $w$ on the root-to-$x$ path. This update requires $O(\log n)$ time as $|\Sigma| = \poly(n)$. 
    
    \item \textbf{Deletion} The procedure to delete is similar to the insertion procedure. To delete $(i,x)$ in the tree, first decrease the value of the $\countSol_x$ and set $\bitVector[i] = 0$. If $\countSol < k$, then, for all $w$ on the root-to-$x$ path which do not have both children $w_0, w_1$ with $\flag_{w_0} = \flag_{w_1} = 1$, set $\flag_w = 0$. This requires $O(\log n)$ time. 
    
    \item \textbf{Query} To check if the tree has $k$ distinct indices with the same $x$, we need only check if $\flag_{\text{root}}=1$, which takes $O(1)$ time.
\end{itemize}

\paragraph{Runtime, error and memory usage} Using the above data-structure, the runtime of Ambainis' algorithm is now $\tO(n^{\frac{k}{k+1}})$ time. The total memory used is $O(n|\Sigma|)$ bits. However, note that at any point of time in any branch of computation Ambainis' walk algorithm stores sets of size $r = O(n^{\frac{k}{k+1}})$. Hence their algorithm with this data structure is a $\tO(n^{\frac{k}{k+1}})$-sparse algorithm. Thus, invoking Theorem~\ref{thm:main} we conclude the following.

\begin{corollary}
There is a bounded-error $\QRAM$ algorithm that computes $k$-Element Distinctness in $\tO(n^{k/{k+1}})$ time using $\tO(n^{k/{k+1}})$ memory qubits.
\end{corollary}

\subsection{Quantum Algorithms for Closest-Pair and related Problems}
\label{sec:AaronsonClosestPair}
The paper of Aaronson et al \cite{Aaronson-ClosestPair-2019} provides quantum algorithms and conditional lower-bounds for several variants of the Closest Pair problem ($\closestPair$).

Let $\Delta(a,b) = \|a-b\|$ denote the Euclidean distance. We then describe the Closest Pair problem under Euclidean distance $\Delta$, but we could have chosen any other metric $\Delta$ in $d$-dimensional space which is strongly-equivalent to the Euclidean distance (such as $\ell_p$ distance, Manhattan distance, $\ell_\infty$, \textit{etc}).

\begin{definition}[Closest Pair ($\CP(n,d)$) problem]
In the $\CP(n,d)$ problem, we are given a list $P$ of $n$ distinct points in $\mathcal{R}^d$, and wish to output a pair $a,b \in P$ with the smallest $\Delta(a,b)$.
\end{definition}

%In this problem and its variants, one is given $n$ points in a $d$-dimensional vector space and wishes to find a pair of points which are closest to each other.  When $d = \mathrm{polylog}(n)$, Aaronson et al prove a lower-bound conditional on a quantum variant of the Strong Exponential-Time Hypothesis. For the constant-dimensional $d = O(1)$ case, $\closestPair$ and its variants have a lower-bound of $\Omega(n^{2/3})$ by reduction from Element Distinctness \cite{Shi-EDlowerBound-2002}. There is a matching query algorithm by Volpato and Moura (CITE), but making such an algorithm time-efficient is non-trivial. Aaronson et al then provide quantum algorithms running in time $\tilde O(n^{2/3})$ for most of $\CP$ variants. 

We may also define a threshold version of $\CP$.

\begin{definition}
In the $\TCP(n,d)$ problem, we are given a set $P = \{p_1, \ldots, p_n\}$ of $n$ points in $\mathbb{R}^d$ and a threshold $\eps \ge 0$, and we wish to find a pair of points $a,b \in P$ such that $\Delta(a,b) \leq \eps$, if such a pair exists.
\end{definition}

For simplicity, so we may disregard issues of representation of the points, we assume that all points are specified using $O(\log n)$ bits of precision. By translation, we can assume that all the points lie in in the integer hypercube $[L]^d$ for some $L = \poly(n)$, and that $\delta \in [L]$, also. 

It is then possible to solve $\CP$ by running a binary search over the (at most $n^2$) different values of $\delta \in \{\Delta(p_i, p_j) \mid i,j \in [n]\}$ and running the corresponding algorithm for $\TCP$. This will add an additional $O(\log n)$ factor to the running time.

The $\TCP(n,d)$ problem is a query problem with certificate complexity $2$. If one is familiar with quantum walks, it should be clear that we may do a quantum walk on the Johnson graph over $n$ vertices, to find a pair with the desired property, by doing $O(n^{2/3})$ queries to the input. Again, if one is familiar with quantum walks, one will realize that, in order to implement this walk efficiently, we must dynamically maintain a set $S \subseteq [n]$, and at each step in the quantum walk, we must be able to add or remove an element $i$ to $S$, and answer a query of the form: does there exist a pair $i,j \in S$ with $\Delta(p_i, p_j) \le \eps$?

The only difficulty, now, is to implement an efficient data structure that can dynamically maintain $S$ in this way, and answer the desired queries, while being time and space efficient. Aaronson et al construct a data-structure which can store a set $S \subseteq [n] \times [L]^d$ of points of size $r$, under efficient insertions and removals, so that one may query at any given time whether there exist two points in $S$ which are $\eps$-close. They do so by first discretizing $[L]^d$ into a hypergrid of width $\eps/\sqrt{d}$, as explained below, and then use a hash table, skip list, and a radix tree to maintain the locations of the points in the hypergrid. 

The presentation of the data structure in the paper is roughly 6 pages long, and one must refer to Ambainis' paper for the error analysis, which is absent from the paper. As we will see, a simple, sparse data structure for the same purpose can be described in less than 2 pages, and then an appeal to Theorem \ref{thm:main} gives us the same result up to log factors. % Note that, though the input is a set of points $(i,x) \in [n]\times [L]^d$ we use a data structure that will store $(i, \fid(x)) \in [n]\times \Sigma$, whenever we need the value of $x$ we invoke the input oracle in our algorithm.

\paragraph{Discretization} We discretize the cube $[L]^d$ into a hypergrid of width $w = \frac{\eps}{\sqrt d}$, and let $\fid(p)$ denote the box containing $p$ in this grid. I.e., we define a function $\fid(p):[L]^d \rightarrow \ZO^{\lceil d \log(L/\eps) \rceil}$ by
\begin{equation}
    \fid(p)=(\lfloor p(1)/w \rfloor, \ldots, \lfloor p(d)/w \rfloor)\tag*{(represented in binary).}
\end{equation}

Let $\Sigma=\ZO^{\lceil d \log(L/\eps) \rceil}$ denote the set of all possible boxes. We say that two boxes $g,g'\in\Sigma$ are neighbours if 
\begin{equation*}
    \sqrt{\sum_{i=1}^d \norm{g(i)-g'(i)}^2} \leq \sqrt{d}.
\end{equation*}
A loose estimate will show there can be at most $(2\sqrt{d}+1)^d$ neighbours for any box.
This method of discretization ensures the following crucial property: 
\begin{observation}[Observation~45 \cite{Aaronson-ClosestPair-2019}]\label{obs:boxes} Let $p, q$ be any two distinct points in $[0,L]^d$, then
\begin{enumerate}
    \item if $\fid(p) = \fid(q)$, then $\Delta(p, q) \leq \eps$, and
    \item if $\Delta(p, q) \leq \eps$, then $\fid(p)$ and $\fid(q)$ are neighbours.
\end{enumerate}
\end{observation}

From Observation \ref{obs:boxes}, it follows that $i, j \in [n]$ exist with $\Delta(p_i,p_j) \le \eps$, if and only if we have one of the following two cases:
\begin{itemize}
    \item Either there is such a pair $i, j$ with $\fid(p_i) = \fid(p_j)$.
    \item Or there is no such pair, and then there must exist two neighbouring boxes $\fid(i)$ and $\fid(j)$, each containing a single point, with $\Delta(p_i, p_j) \le \eps$.
\end{itemize}

\bigskip
We now describe the data structure itself. Let us assume without loss of generality that $n$ is a power of $2$.

\begin{definition}[Data Structure for $\CP$]
\label{def:CPdatastructure}
Let $S \subseteq [n] \times \Sigma$, with $|S|=r$, and such that every $i \in [n]$ appears in at most one pair $(i,x) \in S$. The \textit{closest-pair tree} that represents $S$, denoted by $\tcpT(S)$, is a complete rooted binary tree with $|\Sigma|$ leaves. Each leaf node $x \in \Sigma$ is labeled by a number $\externalCount_x \in \{0,\dots,n\}$, and a prefix-sum tree $P(S_x)$ representing the set $S_x = \{i \in [n] \mid (i,x) \in S\}$. Each internal node $w$ is labeled by a bit $\flag_w \in \ZO$. These labels obey the following rules:
\begin{itemize}
    \item If $|S_x| = 1$, then $\externalCount_x$ is the number of boxes $y \neq x$, which are neighbours of $x$, and which have $|S_y| = 1$ and $\Delta(p_i, p_j) \le \eps$ for the (unique) $j \in S_y$.
    \item If $|S_x| \ge 2$, then $\externalCount_x = 0$.
    \item The $\flag_w=1$ if any of the children $x$ descendants to the internal node $w$ have either $|S_x| \ge 2$ or $|S_x| = 1$ and $\externalCount_x \ge 1$. 
    \end{itemize}
\end{definition}

It follows from the above discussion that there exist two elements $(i,x),(j,y) \in S$ with $\Delta(p_i, p_j) \le \eps$ if and only if $\flag_{\text{root}} = 1$ in $\tcpT(S)$. We now show how to efficiently maintain $\tcpT(S)$ under insertions and removals.

\paragraph{Memory Representation} A $\TCP$ tree is represented in the memory by an array of $|\Sigma| - 1$ bits of memory holding the flags of the internal nodes, followed by $|\Sigma|$ blocks of $n \log n + n$ bits of memory each, holding the labels of the leaf nodes. The blocks appear in the same order as a breadth-first traversal of $\tcpT(S)$. Consequently, for every $S \subseteq [n] \times \Sigma$ there is a corresponding binary string of $|\Sigma|-1 + (n \log n + n)|\Sigma|$ that uniquely encodes $\tcpT(S)$. Crucially, if $|S| = r$, then at most $O(r(\log|\Sigma|+\log n))$ of these bits are $1$. Since $|\Sigma| = L^{O(d)} = \poly(n)$ (recall $d = O(1)$), the encoding is $\tO(r)$-sparse.

\begin{figure}[h]
    \centering
    \scalebox{0.8}{
\begin{tikzpicture}[scale=0.5]
\draw (0,0) node 
  [shape=circle, minimum height=1.5cm, draw=black]
 (Level1){\texttt{flag}};
 
\draw (0,-4.5) node 
  [shape=circle, minimum height=1.5cm, draw=black]
 (Level21){\texttt{flag}};
\draw (8,-4.5) node 
  [shape=circle, minimum height=1.5cm, draw=black]
 (Level22){\texttt{flag}};
\draw (-8,-4.5) node 
  [shape=circle, minimum height=1.5cm, draw=black]
 (Level23){\texttt{flag}};

\draw (2,-9) node 
  [shape=circle, minimum height=1.5cm, draw=black]
 (Level31){\texttt{flag}};
\draw (-2,-9) node 
  [shape=circle, minimum height=1.5cm, draw=black]
 (Level32){\texttt{flag}};
\draw (6,-9) node 
  [shape=circle, minimum height=1.5cm, draw=black]
 (Level33){\texttt{flag}};
\draw (-6,-9) node 
  [shape=circle, minimum height=1.5cm, draw=black]
 (Level34){\texttt{flag}};
\draw (10,-9) node 
  [shape=circle, minimum height=1.5cm, draw=black]
 (Level35){\texttt{flag}};
\draw (-10,-9) node 
  [shape=circle, minimum height=1.5cm, draw=black]
 (Level36){\texttt{flag}};
 
\draw (0,-12.5) node 
  [shape=rectangle, minimum height=0.75cm, minimum width=2cm, draw=black]
 (BitVector0){$x$};
\draw (4,-12.5) node 
  [shape=rectangle, minimum height=0.75cm, minimum width=2cm, draw=black]
 (BitVector1){$...$};
\draw (8,-12.5) node 
  [shape=rectangle, minimum height=0.75cm, minimum width=2cm, draw=black]
 (BitVectorDash){...};
\draw (12,-12.5) node 
  [shape=rectangle, minimum height=0.75cm, minimum width=2cm, draw=black]
 (BitVectorNcube){$\ell_{|\Sigma|}$};
\draw (-4,-12.5) node 
  [shape=rectangle, minimum height=0.75cm, minimum width=2cm, draw=black]
 (BitVectorMinus1){$...$};
\draw (-8,-12.5) node 
  [shape=rectangle, minimum height=0.75cm, minimum width=2cm, draw=black]
 (BitVectorMinusDash){$\ell_2$};
\draw (-12,-12.5) node 
  [shape=rectangle, minimum height=0.75cm, minimum width=2cm, draw=black]
 (BitVectorMinusNcube){$\ell_1$};

\draw [line width=1pt,->] (Level35) -- (BitVectorNcube); 
\draw [line width=1pt,->] (Level36) -- (BitVectorMinusNcube);
\draw [line width=1pt,->] (Level35) -- (BitVectorDash); 
\draw [line width=1pt,->] (Level36) -- (BitVectorMinusDash);
\draw [line width=1pt,->] (Level33) -- (BitVectorDash); 
\draw [line width=1pt,->] (Level34) -- (BitVectorMinusDash);
\draw [line width=1pt,->] (-6,-10.45) -- (BitVectorMinusDash);
\draw [line width=1pt,->] (6,-10.45) -- (BitVectorDash);
\draw [line width=1pt,->] (Level32) -- (BitVectorMinus1);
\draw [line width=1pt,->] (Level32) -- (BitVector0);
\draw [line width=1pt,->] (Level31) -- (BitVectorDash);
\draw [line width=1pt,->] (Level31) -- (BitVector1);

\draw [line width=1pt,->] (Level21) -- (Level31);
\draw [line width=1pt,->] (Level21) -- (Level32);
\draw [line width=1pt,->] (Level23) -- (Level34);
\draw [line width=1pt,->] (Level23) -- (Level36);
\draw [line width=1pt,->] (Level22) -- (Level33);
\draw [line width=1pt,->] (Level22) -- (Level35);
\draw [line width=1pt,dotted, ->] (1,-1) -- (2.5,-2.5);
\draw [line width=1pt,dotted, ->] (-1,-1) -- (-2.5,-2.5);

\draw [line width=1pt, -] (2,-13.25) -- (3.5,-13.75);
\draw [line width=1pt, -] (-2,-13.25) -- (-3.5,-13.75);

\draw (-1.5,-14.5) node 
 [shape=rectangle, minimum height=0.75cm,  minimum width=2cm, draw=black]
 (BitVectorDescription2){$\externalCount_{x}$};
\draw (2,-14.5) node 
  [shape=rectangle, minimum height=0.75cm,  minimum width=1.5cm, draw=black]
 (BitVectorDescription31){$P(S_{x})$};

\end{tikzpicture}
}
    \caption{Data structure for the $\CP$ problem.} 
    \label{fig:TCP}
\end{figure}
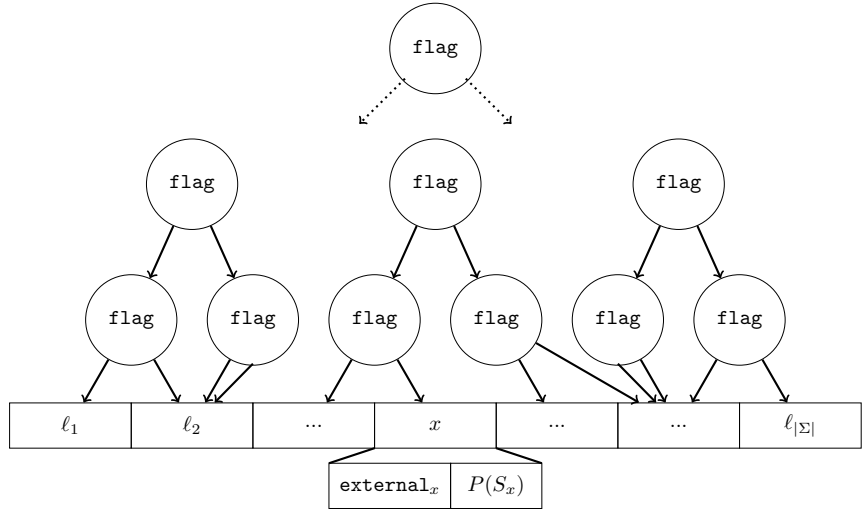

\paragraph{Implementation of data structure operations} It is clear from the definition of $\TCP$ tree and its memory representation that a tree $\tcpT(S)$ represents a set $S \subseteq [n]\times \Sigma$ in a history-independent way. We will now argue that all the required data structure operations take $O(\log n)$ time in the worst case. For every $(i, x) \in [n] \times [L]^d$ there is a corresponding $(i, z) \in [n] \times \Sigma$, with $z=\fid(x)$, stored in the data structure.  
\begin{itemize}
    \item \textbf{Insertion} To insert $(i,x)$ in the tree, first go to the memory location corresponding to leaf $x$. Begin by inserting $i$ in the prefix-sum tree $P(S_x)$. Then three cases arise
    \begin{itemize}
        \item If $|S_x|=1$ then for every neighbour $y$ of $x$ with $|S_y|=1$ do the following: Using the prefix-sum tree at leaf $y$ obtain the only non-zero leaf index $j$ of $P(S_y)$. This operation takes $\log n$ time. Then check if $\Delta(p_i,p_j)\leq \eps$, if yes then increase the values of both $\externalCount_x$ and $\externalCount_y$ by 1. If this caused $\externalCount_y > 0$ then set $\flag_w=1$ for all internal nodes $w$ on the path from leaf $y$ to the root of $\tcpT(S)$.
        
        After going over all neighbours, check if $\externalCount_x \geq 1$, if it is then set $\flag_w=1$ for all internal nodes $w$ on the path from leaf $x$ to the root of $\tcpT(S)$. This process takes at most $(2\sqrt{d}+1)^d\log n$ time as there will be at most $(2\sqrt{d}+1)^d$ neighbours, which is $O(\log n)$ for  $d=O(1)$.
        
        \item If $|S_x|=2$ using the prefix-sum tree $P(S_x)$ obtain the only other non-zero leaf index $i'\neq i$ of $P(S_x)$. Then for all neighbours $y$ of $x$ with $|S_y|=1$ do the following: Using the prefix-sum tree $P(S_y)$ obtain the only non-zero index $j$ of $P(S_y)$. Check if  $\Delta(p_{i'},p_j) \leq \eps$, and if so decrease the value of $\externalCount_y$ by 1. If that results in making $\externalCount_y=0$ then set $\flag_w = 0$ for the parent of $y$, unless the other child $y'$ of the parent of $y$ has $|S_{y'}| \ge 2$ or $\externalCount_{y'} \ge 1$. Likewise, among all the internal nodes $w$ that are on the path from the root to $y$'s parent, update the $\flag_w$ accordingly, i.e., set $\flag_w=1$ if any child $u$ of $w$ has $\flag_u=1$, and otherwise set $\flag_w=0$.
        
        Having done that, set $\externalCount_x=0$ and set $\flag_w=1$ for all internal nodes $w$ from leaf $x$ to the root $\tcpT(S)$. This process also takes $O(\log n)$ time (when $d$ is a constant).
        \item If $|S_x|>2$ then do nothing.
    \end{itemize}

    \item \textbf{Deletion} The procedure to delete is similar to the insertion procedure.
    
    \item \textbf{Query} To check if the tree has a pair $(i,x),(j,y) \in S$ such that $\Delta(p_i,p_j) \leq \eps$, we need only check if $\flag_{\text{root}}=1$, which takes $O(1)$ time.
\end{itemize}

\paragraph{Runtime, error and memory usage} Using the above data-structure, the runtime of this $\TCP$ algorithm is now $\tO(n^{\frac{2}{3}})$ time. The total memory used is $\tO(n|\Sigma|)$ bits. However, note that at any point of time in any branch of computation this algorithm stores sets of size $r = O(n^{\frac{2}{3}})$. Hence their algorithm with this data structure is a $\tO(n^{\frac{2}{3}})$-sparse algorithm. Thus, invoking Theorem~\ref{thm:main} we conclude the following.

\begin{corollary}
There is a bounded-error $\QRAM$ algorithm that computes $\TCP$ in $\tO(n^{2/{3}})$ time using $\tO(n^{2/{3}})$ memory qubits.
\end{corollary}

\subsection{Fine-Grained Complexity via Quantum Walks}
\label{sec:BLPS22}
The authors' own paper \cite{buhrman-3SUMhard-2021} shows that the quantum 3SUM conjecture, which states that there exists no truly sublinear quantum algorithm for 3SUM, implies several other quantum lower-bounds. The reductions use quantum walks together with complicated space-efficient data structures. We had already realized, when writing the paper, that simple yet space-inefficient data structures could be used instead, and included this observation in the paper, so we will not repeat it here. Section 3.1, with the space inefficient sparse data structures, is 4 pages long, whereas section 3.2, with the complicated space efficient data structures, is 12 pages long.

\section{Acknowledgments}
Subhasree Patro is supported by the Robert Bosch Stiftung. Harry Buhrman, Subhasree Patro, and Florian Speelman are additionally supported by NWO Gravitation grants NETWORKS and QSC, and EU grant QuantAlgo. Bruno Loff's research is supported by National Funds through the Portuguese funding agency, FCT - Fundação para a Ciência e a Tecnologia, within project LA/P/0063/2020. This work was supported by the Dutch Ministry of Economic Affairs and Climate Policy (EZK), as part of the Quantum Delta NL programme.

\bibliographystyle{alpha}
\bibliography{CiteHere.bib}

\end{document}